\documentclass[11pt]{article}

\usepackage[a4paper,margin=1in]{geometry}

\usepackage{amsmath,amssymb,amsfonts,amsthm,mathtools}
\usepackage{comment}

\newtheorem{theorem}{Theorem}
\newtheorem{proposition}[theorem]{Proposition}

\theoremstyle{definition}

\theoremstyle{remark}
\newtheorem{remark}[theorem]{Remark}

\DeclareMathOperator{\sech}{sech}

\title{A Semilinear Wave Sector in Force-Free Electrodynamics}
\author{Yafet E. Sanchez Sanchez }
\date{}
\begin{document}

\maketitle

\begin{abstract}
We introduce an ansatz for force-free electrodynamics in Minkowski spacetime under which the nonlinear system reduces to a semilinear scalar wave equation depending on two spacetime variables. This reduction yields explicit time-dependent solutions, including type-changing configurations with finite energy per unit transverse area and a null kink-type example. For traveling-wave solutions in the magnetically dominated regime, the kernel distribution of the field defines minimal field-sheet foliations.
\end{abstract}

\section{Introduction}
Force-free electrodynamics (FFE) describes the dynamics of electromagnetic fields interacting with a highly conducting plasma. In this regime, the electromagnetic field dominates the energy density of the system, and the Lorentz force acting on the plasma vanishes. 

More precisely, FFE on a four-dimensional globally hyperbolic spacetime $(\mathcal{M},g)$ is described by
\begin{align}\label{ffe1}
dF &= 0, \\
*d{*F} &= J, \label{ffe2}\\\label{ffe3}
i_{J^\#} F &= 0,
\end{align}
where $F \in \Omega^2(\mathcal{M})$ is the electromagnetic field tensor, $J^\#:=g^{-1}(J,\cdot)$ is the four current density, and $*$ denotes the Hodge dual. The first two equations are Maxwell's equations with a source, while the third is the force-free constraint expressing the vanishing of the Lorentz force.

The mathematical structure of FFE has been extensively developed over the past decades. Early covariant formulations and geometric interpretations were introduced by Uchida in his foundational works \cite{Uchida1997FFE1,Uchida1997FFE2}. Subsequent progress on time-dependent formulations and hyperbolic structures suitable for analysis and numerical implementation was made by  Komissarov \cite{Komissarov2002FFE}, as well as by  Pfeiffer and MacFadyen \cite{PfeifferMacFadyen2013FFE}, and later by Carrasco and Reula \cite{CarrascoReula2016}.

On the geometric side, significant progress has been made in the construction and classification of exact solutions. Gralla and Jacobson \cite{GrallaJacobson2014} developed a systematic formulation based on exterior calculus, connecting with known solutions from pulsar and black hole magnetospheres. This framework was subsequently extended to general curved spacetimes without symmetry assumptions by Menon \cite{Menon2020ArbitrarySpacetime, Menon2021NonNullFFE, Menon2021Null}. In the stationary axisymmetric setting, Compère, Gralla, and Lupsasca \cite{CompereGrallaLupsasca2016} introduced a foliation-based reduction leading to new classes of solutions, typically governed by elliptic equations. More recently, Adhikari, Menon, and Medvedev constructed solutions in cosmological backgrounds and exhibited configurations with transitions between electromagnetic regimes \cite{AdhikariMenonMedvedev2024FLRW, Adhikari2025ArbitraryGeometries}.

In this article, we study the ansatz 
\[
F = d\Psi \wedge dy + h(\Psi)\, dt \wedge dx
\]
in Minkowski spacetime, where $\Psi = \Psi(t,x)$. Within this class, the force-free equations reduce to a semilinear wave equation in $(1+1)$ dimensions, providing a simple hyperbolic sector in which explicit solutions and their geometric properties can be analyzed. The main results and consequences of this reduction are summarized below. 
\subsection{Main results}

The main result of this paper is that the ansatz
\[
F=d\Psi\wedge dy+h(\Psi)\,dt\wedge dx,
\qquad \Psi=\Psi(t,x),
\]
reduces the force-free equations in Minkowski spacetime to
\[
\bigl(\Psi_{tt}-\Psi_{xx}+h(\Psi)h'(\Psi)\bigr)\,d\Psi=0,
\]
and hence, on the open set where \(d\Psi\neq 0\), to a semilinear wave equation in \(1+1\) dimensions (Theorem~1). This identifies a simple hyperbolic sector of force-free electrodynamics in which explicit time-dependent solutions can be constructed and analyzed.

We then develop two main consequences of this reduction. First, in the Klein--Gordon case \(h(\Psi)=m\Psi\), the ansatz yields explicit smooth solutions with finite energy per unit transverse area exhibiting transitions between magnetically dominated, null, and electrically dominated regimes. In addition, we also present a nonlinear example based on a sine--Gordon kink, which gives rise to a null force-free configuration. Second, for traveling-wave solutions, the kernel distribution $\ker F$ defines a foliation whose distinguished generator is geodesic; in the magnetically dominated region, the corresponding field sheets are minimal submanifolds (Theorem~\ref{minimal}).

The paper is organized as follows. In Section~2 we introduce the ansatz and derive the scalar reduction. In Section~3 we present explicit solutions and analyze the electromagnetic regimes they induce. In Section~4 we study the associated kernel foliation and prove the minimality result for traveling waves. In Section~5 we conclude with remarks on extensions to curved spacetimes and on a restricted pressureless fluid interpretation.

\section{A scalar reduction of the force-free equations}

We work on $(3+1)$-dimensional Minkowski spacetime $(\mathbb{R}^{1,3},g)$ with coordinates
\[
(t,x,y,z)
\]
and metric
\[
g=-dt^2+dx^2+dy^2+dz^2.
\]

\begin{theorem}\label{thm:reduction}
Let $\Psi=\Psi(t,x)$ and let $h:\mathbb{R}\to\mathbb{R}$ be smooth. For
\[
F=d\Psi\wedge dy + h(\Psi)\,dt\wedge dx,
\]
the force-free equations are equivalent to
\[
\bigl(\Psi_{tt}-\Psi_{xx}+h(\Psi)h'(\Psi)\bigr)\,d\Psi=0.
\]
\end{theorem}

\begin{proof}
We will show that Eq. \eqref{ffe1}, \eqref{ffe2} and \eqref{ffe3} are satisfied.

We have that
\begin{equation*}
d\Psi=\Psi_t\,dt+\Psi_x\,dx,
\end{equation*}
so the two-form $F$ takes the form
\begin{equation*}
F
=
\Psi_t\,dt\wedge dy
+
\Psi_x\,dx\wedge dy
+
h(\Psi)\,dt\wedge dx.
\end{equation*}
Hence the only nonvanishing components of $F_{\mu\nu}$ are
\begin{equation*}
F_{tx}=h(\Psi),\qquad
F_{ty}=\Psi_t,\qquad
F_{xy}=\Psi_x.
\end{equation*}

We first verify Eq.\eqref{ffe1}. Since
\begin{equation*}
F=d\Psi\wedge dy+h(\Psi)\,dt\wedge dx,
\end{equation*}
we have
\begin{equation*}
dF=d(d\Psi\wedge dy)+d\bigl(h(\Psi)\,dt\wedge dx\bigr).
\end{equation*}
The first term vanishes since $d^2=0$ and $dy$ is closed. For the second term,
\[
d\bigl(h(\Psi)\,dt\wedge dx\bigr)
=h'(\Psi)\,d\Psi\wedge dt\wedge dx=0,
\]
as $\Psi=\Psi(t,x)$. Hence $dF=0$.

Next we compute the current, which is the content of Eq.\eqref{ffe2}, in coordinates it is given by $J^\mu=\partial_\nu F^{\nu\mu}$.
A direct computation gives
\begin{equation*}
J^\mu=
\bigl(
h'(\Psi)\Psi_x,\,
-h'(\Psi)\Psi_t,\,
-\Psi_{tt}+\Psi_{xx},\,
0
\bigr).
\end{equation*}

Now impose the force-free condition, given by Eq. \eqref{ffe3}, $F_{\mu\nu}J^\nu=0$.

The $t$- and $x$-components give
\begin{align}
\Psi_t\Bigl(-\Psi_{tt}+\Psi_{xx}-h(\Psi)h'(\Psi)\Bigr)&=0,\label{eq:ff1}\\
\Psi_x\Bigl(-\Psi_{tt}+\Psi_{xx}-h(\Psi)h'(\Psi)\Bigr)&=0.\label{eq:ff2}
\end{align}
The $y$- and $z$-components vanish identically. Combining \eqref{eq:ff1} and \eqref{eq:ff2}, we obtain
\begin{equation*}
\Bigl(-\Psi_{tt}+\Psi_{xx}-h(\Psi)h'(\Psi)\Bigr)d\Psi=0.
\end{equation*}
\end{proof}

\begin{remark}
On the open set where $(\Psi_t,\Psi_x)\neq(0,0)$,
the force-free condition is equivalent to $-\Psi_{tt}+\Psi_{xx}-h(\Psi)h'(\Psi)=0$.
If $d\Psi=0$, then \eqref{eq:ff1} and \eqref{eq:ff2} are trivially satisfied.
\end{remark}

\begin{remark}
Within the ansatz-defined sector, the force-free equations reduce to a semilinear 
wave equation, for which local well-posedness follows from standard theory 
(see, e.g., \cite[Proposition~9.12]{Ringstrom2009Cauchy} and 
\cite[Chapter~3]{Tao2006}). In the special case $h(\Psi) = m\Psi$, the reduction 
yields the Klein--Gordon equation, whose smooth initial data are known to give rise 
to globally smooth solutions \cite[Proposition~8.6]{Ringstrom2009Cauchy}.
\end{remark}

\begin{remark}
For the choice $h(\Psi)=m\Psi$, if $\Psi_1$ and $\Psi_2$ are solutions, then the corresponding fields $F(\Psi_1)$ and $F(\Psi_2)$ are force-free. By linearity of the Klein--Gordon equation, $\Psi_1+\Psi_2$ is also a solution, and hence $F(\Psi_1+\Psi_2)$ is again force-free. This identifies a superposition property within the Klein–Gordon subclass of the present ansatz.
\end{remark}

 \subsection{Electromagnetic fields}

For later use, we record the electric and magnetic fields associated with the ansatz. Using the convention
\begin{equation*}
E_i=F_{ti},\qquad F_{ij}=\epsilon_{ijk}B^k,
\end{equation*}
we obtain
\begin{equation*}
E=(h(\Psi),\,\Psi_t,\,0),
\qquad
B=(0,\,0,\,\Psi_x).
\end{equation*}
In particular,
\begin{equation*}
E\cdot B=0.
\end{equation*}

\subsubsection{Magnetic, null, and electric regimes}

The electromagnetic invariant distinguishing the magnetic, null, and electric regimes for the present ansatz becomes
\begin{equation}\label{regime}
\Delta=B^2-E^2=\Psi_x^2-\Psi_t^2-h(\Psi)^2.
\end{equation}

Thus:
\begin{equation*}
B^2-E^2>0
\quad\Longleftrightarrow\quad
\Psi_x^2-\Psi_t^2>h(\Psi)^2,
\end{equation*}
corresponding to the magnetically dominated regime,
\begin{equation*}
B^2-E^2=0
\quad\Longleftrightarrow\quad
\Psi_x^2-\Psi_t^2=h(\Psi)^2,
\end{equation*}
corresponding to the null regime, and
\begin{equation*}
B^2-E^2<0
\quad\Longleftrightarrow\quad
\Psi_x^2-\Psi_t^2<h(\Psi)^2,
\end{equation*}
corresponding to the electrically dominated regime.

\subsubsection{Energy density and Poynting vector}

The electromagnetic energy density and Poynting vector are
\[
\mathcal E
=
\frac12(|E|^2+|B|^2)
=
\frac12(\Psi_t^2+\Psi_x^2+h(\Psi)^2),
\]

\[
S=E\times B
=
(-\Psi_t\Psi_x,\;h(\Psi)\Psi_x,\;0).
\]

\section{Type-changing and kink-type configurations}\label{sec:examples}
In this section, we construct a solution exhibiting a transition between magnetic and electric dominance, and then give a nonlinear example based on a sine--Gordon kink.

\subsection{A type-changing solution with finite energy per unit transverse area}

We construct a type-changing solution with finite energy per unit transverse area in two steps. First, we analyze monochromatic waves, and then we localize them to obtain the desired solution.

\emph{Step 1. Monochromatic Wave Analysis}

Consider
\[
\Psi(t,x)=a\cos(\xi_0 x-\tfrac{\pi}{4})\cos(\omega_0 t),
\qquad
\omega_0=\sqrt{\xi_0^2+m^2}.
\]
Then \(\Psi\) solves the Klein--Gordon equation
\[
\Psi_{tt}-\Psi_{xx}+m^2\Psi=0.
\]
At the origin,
\[
\Psi(t,0)=\frac{a}{\sqrt2}\cos(\omega_0 t),\qquad
\Psi_x(t,0)=\frac{a\xi_0}{\sqrt2}\cos(\omega_0 t),\qquad
\Psi_t(t,0)=-\frac{a\omega_0}{\sqrt2}\sin(\omega_0 t).
\]
Hence
\[
\Delta(t,0):=\Psi_x(t,0)^2-\Psi_t(t,0)^2-m^2\Psi(t,0)^2
=
\frac{a^2}{2}(\xi_0^2-m^2)\cos^2(\omega_0 t)
-\frac{a^2}{2}(\xi_0^2+m^2)\sin^2(\omega_0 t).
\]
In particular, if \(\xi_0>m\), then
\[
\Delta(0,0)=\frac{a^2}{2}(\xi_0^2-m^2)>0,
\]
while at
\[
t_*=\frac{\pi}{2\omega_0}
\]
one has
\[
\Delta(t_*,0)=-\frac{a^2}{2}(\xi_0^2+m^2)<0.
\]
Therefore the electromagnetic type at the origin changes from magnetic to electric.

\emph{Step 2. Construction of a solution with finite energy per unit transverse area}

To turn the above monochromatic example into a finite energy solution per unit transverse area, let
\[
\Psi^{\mathrm{mono}}(t,x)=a\cos(\xi_0 x-\tfrac{\pi}{4})\cos(\omega_0 t),
\qquad
\omega_0=\sqrt{\xi_0^2+m^2},
\]
and fix a time \(T>0\) such that the sign change at the origin occurs on the interval \([0,T]\); for instance one may take
\[
T=\frac{\pi}{2\omega_0}.
\]
Choose \(R>T\), and let \(\chi\in {\mathcal{S}}^\infty(\mathbb R)\) be a rapidly decaying function satisfying
\[
\chi(x)=1 \quad \text{for } |x|\le R.
\]
Now define initial data
\[
\Psi_0(x)=\chi(x)\,a\cos(\xi_0 x-\tfrac{\pi}{4}),
\qquad
\Psi_1(x)=0.
\]
Then
\[
\Psi_0\in H^1(\mathbb R),\qquad \Psi_1\in L^2(\mathbb R),
\]
so the corresponding Klein--Gordon solution \(\Psi\) has finite energy.

Moreover, on the interval \(|x|\le R\), the localized initial data agree exactly with the monochromatic data. Since the Klein--Gordon equation has finite speed of propagation, the solution \(\Psi\) agrees with \(\Psi^{\mathrm{mono}}\) in the domain of dependence of \([-R,R]\), namely in the region
\[
\{(t,x): 0\le t\le T,\ |x|\le R-t\}.
\]
Because \(R>T\), the origin \((t,0)\) for \(0\le t\le T\) lies in this region. Therefore
\[
\Psi(t,0)=\Psi^{\mathrm{mono}}(t,0)
\qquad\text{for all } 0\le t\le T.
\]
In particular, the invariant
\[
\Delta(t,0)=\Psi_x(t,0)^2-\Psi_t(t,0)^2-m^2\Psi(t,0)^2
\]
coincides with the monochromatic invariant on \([0,T]\). Hence, if \(\xi_0>m\),
\[
\Delta(0,0)>0,
\qquad
\Delta\!\left(\frac{\pi}{2\omega_0},0\right)<0.
\]
Thus the smooth initial data above produce a field with finite energy per unit transverse area and for which the electromagnetic type at the origin changes from magnetic to electric.

\begin{remark}
For monochromatic solutions, the frequency $\xi_0$ biases the predominant regime 
($\xi_0\ll m$ favors electric, $\xi_0\gg m$ favors magnetic), but does not prevent 
the occurrence of type change. More generally, the ansatz allows for control of the electromagnetic regime through the initial data. For example, if the initial data satisfy $\Psi(0,x)=0$ and $\Psi_t(0,x)\neq 0$ then the field is initially electrically dominated, 
while initial data with $\Psi_t(0,x)=0$ and $(\Psi_x(0,x))^2>m^2(\Psi(0,x))^2$ correspond to a magnetically dominated field.
\end{remark}

\subsection{A sine--Gordon kink example}
 We illustrate the ansatz with the simplest traveling kink solution of the sine--Gordon equation.

For the sine--Gordon equation
\[
\Psi_{tt}-\Psi_{xx}+\sin\Psi=0,
\]
we have
\[
h(\Psi)=2\sin(\Psi/2).
\]

The traveling kink solution \cite[Eq.2]{DmitrievKevrekidis2014}
\[
\Psi(t,x)=4\arctan
\left(
e^{\frac{x-vt}{\sqrt{1-v^2}}}
\right),
\qquad |v|<1,
\]
gives

\[
F=
\frac{2}{\sqrt{1-v^2}}\mathrm{sech}(\xi)
(dx-vdt)\wedge dy
+
2\,\mathrm{sech}(\xi)\,dt\wedge dx.
\]

The corresponding electromagnetic fields are
\[
E=(2\,\mathrm{sech}\xi,-\tfrac{2v}{\sqrt{1-v^2}}\mathrm{sech}\xi,0),
\qquad
B=(0,0,\tfrac{2}{\sqrt{1-v^2}}\mathrm{sech}\xi).
\]

A direct computation shows
\[
B^2-E^2=0,
\]
so the sine--Gordon kink produces a null force-free configuration.

Furthermore, the energy density is given by 
\[
{\cal{E}}=\frac{4}{1-v^2}\,\mathrm{sech}^2\xi,
\]
which is localized, with a profile that decays rapidly away from $\xi=0$.
\section{Foliations}\label{sec:foliation}

\subsection{The kernel distribution $\ker F$ and the vector field $K$}
In this section, we connect the ansatz introduced above with geometric constructions of force-free solutions, particularly geometric and foliation-based approaches to force-free fields (see e.g., \cite{GrallaJacobson2014, Menon2021NonNullFFE}). 

\begin{proposition}\label{foliations}
Let
\[
F=d\Psi\wedge dy+h(\Psi)\,dt\wedge dx,
\]
with \(F\neq 0\) on an open set \(N\). Then
\[
\ker F=\{X\in TN:\iota_XF=0\}=\operatorname{span}\{\partial_z,K\},
\]
where
\[
K=\Psi_x\partial_t-\Psi_t\partial_x+h(\Psi)\partial_y.
\]
Moreover, \(\ker F\) is an involutive rank-two distribution on \(N\), and therefore defines a two-dimensional foliation of \(N\).
\end{proposition}

\begin{proof}
\hspace{0.1cm}
 That $\partial_z\in \text{ker}F$ follows since $F$ does not depend on $dz$.

Notice that
\[
i_K(d\Psi\wedge dy)
=
(i_Kd\Psi)\,dy-d\Psi\,(i_Kdy).
\]

Hence
\begin{equation}\label{geo1}
i_K(d\Psi\wedge dy)=-h(\Psi)\,d\Psi.
\end{equation}

Next,
\[
i_K\bigl(h(\Psi)\,dt\wedge dx\bigr)
=
h(\Psi)\,i_K(dt\wedge dx).
\]
which gives 
\[
i_K(dt\wedge dx)
=
\Psi_x\,dx-(-\Psi_t)\,dt
=
\Psi_x\,dx+\Psi_t\,dt
=
d\Psi.
\]
Therefore
\begin{equation}\label{geo2}
i_K\bigl(h(\Psi)\,dt\wedge dx\bigr)
=
h(\Psi)\,d\Psi.
\end{equation}

Combining Eq.\eqref{geo1} and Eq. \eqref{geo2} gives 
\[
i_KF
=
-h(\Psi)\,d\Psi+h(\Psi)\,d\Psi
=
0.
\]

Thus
\[
{i_KF=0}
\]
which shows $K\in \text{ker}F$.

    Since $\text{ker}F$ is two dimensional on the open set where $F\neq0$,  $\partial_z$ and $K$ form a basis.

 The rest of the proposition follows from the fact that since  $\Psi$ and $h(\Psi)$ depend only on $(t,x)$, we have
$[\partial_z,K]=0$, and therefore using Frobenious theorem \cite[Theorem 2.2.26]{AbrahamMarsden2008} gives the result.
\end{proof}

We now analyze geometric properties of the distinguished vector field $K$,
which will be used to establish the minimality of the associated foliation.

\begin{proposition}\label{K}
    Let $$K=\Psi_x\partial_t-\Psi_t\partial_x+h(\Psi)\partial_y.$$
Then the following hold:
    
\begin{enumerate}

\item $g(K,K)=-(|B|^2-|E|^2)$
\item The integral curves of $K$ are affinely parametrized geodesics if and only if
\[
{
\Psi_x\Psi_{tx}-\Psi_t\Psi_{xx}=0,\qquad
\Psi_t\Psi_{tx}-\Psi_x\Psi_{tt}=0.
}
\]
\end{enumerate}

\end{proposition}

\begin{proof}
    \begin{enumerate}
        \item A direct computation combined with Eq. \eqref{regime} gives the result.

\item 
 
We have 
\[
\nabla_K K = K^\nu \partial_\nu K^\mu \,\partial_\mu.
\]

The components of $K$ are
\[
K^t=\Psi_x,\qquad K^x=-\Psi_t,\qquad K^y=h(\Psi),\qquad K^z=0.
\]

Hence,
\[
\nabla_K K
=
(\Psi_x\Psi_{xt}-\Psi_t\Psi_{xx})\,\partial_t
-
(\Psi_x\Psi_{tt}-\Psi_t\Psi_{tx})\,\partial_x.
\]
 
\end{enumerate}
\end{proof}
\begin{remark}
    Thus in the magnetically dominated region $|B|^2-|E|^2>0$ the vector
field $K$ is timelike and the leaves are timelike two--surfaces.
\end{remark}

\begin{remark}
The condition
\[
\nabla_K K = \alpha K
\]
for some scalar function $\alpha$ forces $\alpha h(\Psi)=0$. Thus, in the nontrivial force-free case $h(\Psi)\neq 0$, one must have $\alpha=0$, and hence $K$ is affinely geodesic. Non-affinely parametrized geodesics can occur only in the degenerate case $h(\Psi)=0$, corresponding to vacuum Maxwell fields.
\end{remark}

\subsection{Minimal Foliations}
We now show that for traveling-wave solutions, $K$ is geodesic and the corresponding foliation is minimal.

\begin{proposition}\label{geodesic}
Let $\Psi(t,x)=f(x-vt)$ be a traveling-wave solution. Then the associated kernel vector field
\[
K=\Psi_x\,\partial_t-\Psi_t\,\partial_x+h(\Psi)\,\partial_y
\]
satisfies
\[
\nabla_K K=0.
\]
\end{proposition}

\begin{proof}
Let
\[
\xi=x-vt.
\]

Therefore
\[
K=f'(\xi)\,\partial_t+v f'(\xi)\,\partial_x+h(f(\xi))\,\partial_y.
\]

From Proposition \ref{K}, the condition $\nabla_K K=0$ is equivalent to
\[
\Psi_x\Psi_{tx}-\Psi_t\Psi_{xx}=0,
\qquad
\Psi_t\Psi_{tx}-\Psi_x\Psi_{tt}=0.
\]
Substituting the traveling-wave expressions, we obtain that both conditions hold, and therefore
\[
\nabla_K K=0.
\]
Thus the integral curves of $K$ are affinely parametrized geodesics.
\end{proof}

\begin{remark}\label{deltaflow}
For traveling-wave solutions $\Psi(t,x)=f(x-vt)$, the invariant $\Delta$
depends only on $\xi=x-vt$, and hence $K(\Delta)=0$. This shows the field-type is constant along the flow.
\end{remark}

We now show that the geodesic condition implies minimality of the foliation.

\begin{theorem}\label{minimal}
Let $\Psi=f(x-vt)$ be a traveling-wave solution and assume $|B|^2-|E|^2 > 0$.
Then the integral surfaces of $\ker F$ are minimal submanifolds.
\end{theorem}

\begin{proof}
     
In the non--null  region, define
\[
\Delta = |B|^2 - |E|^2 > 0,
\qquad
e_0 = \frac{K}{\sqrt{\Delta}}, 
\qquad
e_1 = \partial_z.
\]
Then $\{e_0,e_1\}$ is an orthonormal frame for $T(\ker F)$.

The mean curvature vector $\vec{H}$ of the corresponding leaves is given by \cite[p. 101]{ONeill1983}
\[
\vec{H} = \frac{1}{2}\bigl(\nabla_{e_0} e_0 + \nabla_{e_1} e_1\bigr)^\perp,
\]
where $\perp$ denotes the projection to the normal bundle.

Moreover, we have 
\[
\vec{H} = \frac{1}{2}(\nabla_{e_0} e_0)^\perp,
\]
since $\partial_z$ is parallel in flat spacetime.

Now 
\[
\nabla_{e_0} e_0
= \nabla_{\frac{K}{\sqrt{\Delta}}} \left( \frac{K}{\sqrt{\Delta}} \right)
= \frac{1}{\Delta} \nabla_K K
- \frac{1}{2\Delta^2} K(\Delta)\, K.
\]

Since $K$ is tangential to $\text{ker} F$ and Remark \ref{deltaflow}, we have 
\[(\nabla_{e_0} e_0)^\perp=\frac{1}{\Delta} \nabla_K K\]

Using Proposition \ref{geodesic}, we obtain 
\[
H=0,
\]

which is the minimality condition. 
\end{proof}

Thus traveling-wave solutions generate geodesic flows within the distributions defined by $\ker F$, and the associated two dimensional field sheets are minimal submanifolds of spacetime.

\begin{remark}
    For the choice $h(\Psi)=m\Psi$, the profile
\[
f(\xi)=\sinh(\lambda\xi),
\qquad
\text{with } \xi=x-vt, \lambda=\frac{m}{\sqrt{1-v^2}},
\]
yields an everywhere magnetically dominated force--free field. Furthermore, by Theorem \ref{minimal}, the foliation is minimal.
\end{remark}

\section{Discussion}
We conclude with two limited extensions of the flat-space construction and with a restricted fluid interpretation in the magnetically dominated regime. These remarks indicate possible directions beyond the Minkowski setting.

\subsection{Curved Spacetimes}

We record two limited observations in curved backgrounds. First, by conformal invariance, flat-space solutions yield force-free fields on conformally flat spacetimes. Second, for a restricted product-type class of metrics, the ansatz again leads to a semilinear scalar reduction. 

 \subsubsection{Conformally flat spacetimes}

A first generalization follows from the conformal invariance of the force-free equations, i.e.  
any solution of the flat-space reduction immediately yields a force-free field 
on any conformally flat spacetime.

To be precise, let $\tilde{g}_{ab} = \Omega^2 g_{ab}$. If $F_{ab}$ is a 
solution of Maxwell's equations with current $j$ in $g_{ab}$, then the field 
$\tilde{F}_{ab} = F_{ab}$ satisfies Maxwell's equations \eqref{ffe1}--\eqref{ffe2} 
in $\tilde{g}_{ab}$, with rescaled current $\tilde{j} = \Omega^{-2} j$ 
(see \cite[Eqs.~2.1, 2.9]{CoteFaraoniGiusti2019}). The force-free condition 
is preserved:
\begin{equation*}
    \tilde{F}_{ab}\tilde{j}^{a} 
    = F_{ab}\,\tilde{g}^{cb}\tilde{j}_{c} 
    = F_{ab}\,\Omega^{-2}g^{cb}\,\Omega^{-2}j_{c} 
    = \Omega^{-4}F_{ab}j^{b} 
    = 0.
\end{equation*}

Furthermore, under a conformal rescaling \cite[Eq.2.6]{CoteFaraoniGiusti2019}
\[
\tilde F^{\mu\nu}\tilde F_{\mu\nu}=\Omega^{-4}F^{\mu\nu}F_{\mu\nu},
\]
and the energy density transforms as \cite[Eq.5.4]{CoteFaraoniGiusti2019}
\[
\tilde {\cal{E}}=\frac12\bigl(|\tilde E|^2+|\tilde B|^2\bigr)=\Omega^{-2}{\cal{E}}.
\]
In particular, the electromagnetic character is preserved: type-changing solutions remain type-changing, and the kink configuration remains null. Since $\sech^2$ exhibits exponential decay, the kink profile remains localized provided the conformal factor $\Omega$ is bounded, and therefore the corresponding energy density remains localized.

However, the solutions that induced minimal foliations are not minimal in the conformal spacetime since the mean curvature vector in the conformal spacetime is given by \cite[Theorem 2.1]{GaoZhou2024PMC}
$$\tilde{\vec H}
= \Omega^{-2}\Big(\vec H - (n-1)\, (\nabla^\perp \log \Omega)\Big)$$

These solutions, however, do not arise from a scalar reduction 
directly in the curved background, since  the scalar $\Psi$ still satisfies the flat-space 
wave equation, and the reduction does not reflect the intrinsic geometry of 
$\tilde{g}_{ab}$.

\subsubsection{Product-Type Spacetimes} 

We now present a generalization of the reduction to a broader class of metrics. The key ingredients are a product structure in which the spacetime splits into a two-dimensional Lorentzian base and transverse directions, together with the existence of a distinguished transverse $1$-form that is closed, co-closed, and of constant unit norm. These properties are precisely what allow the force-free system, under the ansatz, to reduce to a scalar semilinear wave equation on the Lorentzian base.

\begin{proposition} \label{curve}
Let
\[
(M,\tilde{g})=\bigl(N\times \mathbb R_y\times I_z,\; g_{ab}(\tilde{x})\,dx^a dx^b+dy^2+Q(z)\,dz^2\bigr),
\]
where $g$ is a Lorentzian metric on $N$ and $Q(z)>0$ is smooth.
Let $\Psi=\Psi(x)$ depend only on the base variables $\tilde{x}=(x^0,x^1)$, and let
\[
F=d\Psi\wedge dy+h(\Psi)\,\mathrm{vol}_g,
\]
where $\mathrm{vol}_g$ is the Lorentzian volume form of $g$.
Then $dF=0$, and the force-free equations are equivalent to
\[
(\Box_g\Psi-h(\Psi)h'(\Psi))\,d\Psi=0.
\]
In particular, on the open set where $d\Psi\neq 0$,
\[
\Box_g\Psi=h(\Psi)h'(\Psi).
\]
\end{proposition}

\begin{proof}
Let $\tilde x=(x^0,x^1)$ be local coordinates on $N$. Since $\Psi=\Psi(\tilde x)$,
\[
d\Psi=\Psi_a\,dx^a,
\qquad
F=\Psi_a\,dx^a\wedge dy+h(\Psi)\,\mathrm{vol}_g.
\]

We first verify $dF=0$. Since $dy$ is closed and $d^2=0$,
\[
d(d\Psi\wedge dy)=0.
\]
Moreover,
\[
d\bigl(h(\Psi)\mathrm{vol}_g\bigr)
= h'(\Psi)\,d\Psi\wedge \mathrm{vol}_g + h(\Psi)\,d(\mathrm{vol}_g)=0,
\]
since $\mathrm{vol}_g$ is closed and $d\Psi\wedge \mathrm{vol}_g=0$ in dimension two.

Next, we compute the current $$J^\mu=\nabla_\nu F^{\nu\mu}=
\frac{1}{\sqrt{|\det \tilde g|}}
\partial_\nu\!\left(\sqrt{|\det \tilde g|}\,F^{\nu\mu}\right),$$

Therefore, for the $y$-component, since $F_{ay}=\Psi_a$ and $F^{ay}=g^{ab}\Psi_b$,

\[
J^y
=
\frac{1}{\sqrt{|\det g|}\sqrt{Q}}
\partial_a\!\left(\sqrt{|\det g|}\sqrt{Q}\, g^{ab}\Psi_b\right)
=
\frac{1}{\sqrt{|\det g|}}
\partial_a\!\left(\sqrt{|\det g|}\, g^{ab}\Psi_b\right)
=
\Box_g\Psi.
\]

For the base components, since $F_{ab}=h(\Psi)(\mathrm{vol}_g)_{ab}$, it follows that $F^{ba}=h(\Psi)(\mathrm{vol}_g)^{ba}$.

Hence
\[
J^a
=
\nabla_b F^{ba}
=
\nabla_b\!\left(h(\Psi)(\mathrm{vol}_g)^{ba}\right)
=
h'(\Psi)\Psi_b(\mathrm{vol}_g)^{ba}
+
h(\Psi)\nabla_b(\mathrm{vol}_g)^{ba}.
\]
Since the volume form is parallel on $(N,g)$, we have $\nabla_b(\mathrm{vol}_g)^{ba}=0$, and therefore
\[
J^a = h'(\Psi)\Psi_b(\mathrm{vol}_g)^{ba}.
\]

Finally, since $F^{\nu z}=0$, we obtain
\[
J^z=0.
\]
Thus
\[
J^y=\Box_g\Psi,\qquad
J^a=h'(\Psi)\Psi_b(\mathrm{vol}_g)^{ba},\qquad
J^z=0.
\]

Imposing the force-free condition $F_{\mu\nu}J^\nu=0$, the $a$-components give
\[
h(\Psi)(\mathrm{vol}_g)_{ab}J^b+\Psi_a\,\Box_g\Psi=0.
\]
Substituting the expression for $J^b$ and using
\[
(\mathrm{vol}_g)_{ab}(\mathrm{vol}_g)^{cb}=-\delta_a^{\,c},
\]
we obtain
\[
(\Box_g\Psi-h(\Psi)h'(\Psi))\Psi_a=0,
\]
that is,
\[
(\Box_g\Psi-h(\Psi)h'(\Psi))\,d\Psi=0.
\]
The $y$-component vanishes identically since it involves the contraction of a symmetric tensor $\Psi_a\Psi_b$ with the antisymmetric tensor $(\mathrm{vol}_g)^{ba}$.

The $z$-components vanish identically, so the result follows.
\end{proof}

\begin{remark}
      There are two main obstructions to include spherical symmetric spacetimes in Proposition \ref{curve}. The first obstruction is that such spacetimes are warped products breaking the product-type setting. The second is that the $S^2$ does not admit a constant norm one-form.  
\end{remark}

\subsection{Type-changing solutions and coupling to a pressureless charged fluid}

We briefly record a restricted single-fluid interpretation of the ansatz in the magnetically
dominated regime. This provides a restricted kinematic realization of the force-free current by a pressureless single fluid, rather than a general plasma model.

Assume that the fluid velocity is tangent to the field sheets, i.e.
\[
u\in\ker F.
\]
Consider a pressureless charged fluid with current
\[
J^\mu=qn\,u^\mu,
\]
where \(q\) is the particle charge, \(n\) the number density, and \(u\) a unit timelike four-velocity.
The coupled system is 
\[
dF=0,\qquad *d*F=J,\qquad \nabla_\mu(nu^\mu)=0,
\]
together with
\begin{equation}\label{fluid}
u^\mu\nabla_\mu u^\nu=\frac{q}{m_p}F^\nu{}_\lambda u^\lambda.
\end{equation}

since by assumption $u\in \ker F$ that implies $$u^\mu\nabla_\mu u^\nu=0.$$

For
\[
F=d\Psi\wedge dy+h(\Psi)\,dt\wedge dx,
\]
the Maxwell current is
\[
J=h'(\Psi)\Psi_x\,\partial_t-h'(\Psi)\Psi_t\,\partial_x+(-\Psi_{tt}+\Psi_{xx})\,\partial_y,
\]
and
\[
\ker F=\operatorname{span}\{\partial_z,K\},\qquad
K=\Psi_x\partial_t-\Psi_t\partial_x+h(\Psi)\partial_y.
\]
If \(u\in\ker F\), then
\[
u=a\,\partial_z+b\,K.
\]
Since the current has no \(\partial_z\)-component, compatibility with \(J^\mu=qn\,u^\mu\)
forces \(a=0\), so the fluid velocity must be aligned with \(K\). Using the reduced equation,
\[
J=h'(\Psi)\,K.
\]In the magnetically dominated region
\[
\Delta = \Psi_x^2 - \Psi_t^2 - h(\Psi)^2 > 0,
\]
the vector field $K$ is timelike. Restricting to the region where $K$ is future-directed, we define
\[
u = \frac{K}{\sqrt{\Delta}}.
\]
More generally, one may take $u = \pm K/\sqrt{\Delta}$, with the sign fixed by requiring consistency with $J^\mu = qn\,u^\mu$.

Since the Lorentz force vanishes, Eq.\eqref{fluid} reduces to the geodesic equation for the vector $u$.
For traveling-wave solutions, Proposition~\ref{geodesic} gives \(\nabla_KK=0\), and since \(\Delta\) depends
only on \(x-vt\), one also has \(K(\Delta)=0\) (Remark \ref{deltaflow}). Hence the normalized flow \(u=K/\sqrt{\Delta}\)
is geodesic.

This interpretation breaks down at null transition points, where \(\Delta\to0\). Indeed,
\(K\) becomes null and the normalization \(u=K/\sqrt{\Delta}\) ceases to be defined, while
\[
qn=h'(\Psi)\sqrt{\Delta}\to0
\]
and the force-free current \(J=h'(\Psi)K\) remains finite. Thus the force-free current remains
well defined, but the pressureless single-fluid interpretation degenerates at the transition. Notice that in this interpretation, the fluid variables are entirely determined by the electromagnetic field and the construction does not introduce independent matter degrees of freedom.

The magnetically dominated region is where the force-free approximation 
is expected to hold \cite{GrallaJacobson2014, Komissarov2002FFE, 10.1093/mnras/stac2720}. This is 
reflected concretely in the single-fluid interpretation admitted by the example above. The ansatz therefore provides a tractable setting 
in which the relation between electromagnetic regime change and the failure 
of a fluid interpretation can be made explicit. Extensions to more general multi-particle models with independent matter degrees of freedom are left for future work.

\section*{Data Availability}

The data that support the findings of this study are available from the authors upon reasonable request.

\section*{Conflict of Interest}

The authors declare that they have no conflict of interest.
\newpage
\bibliographystyle{plain}
\bibliography{biblio}

{\sc Yafet E. Sanchez Sanchez\\ 
Center for Relativity and Cosmology\\ Troy University\\ Troy, AL 36082}\\
{\tt {sanchezy@troy.edu}} 
\medskip
\end{document}